\documentclass[11pt,a4paper,reqno]{amsart} 

\usepackage{mathtools}
\usepackage{enumitem}
\usepackage{graphicx}
\usepackage[foot]{amsaddr}
\usepackage[margin=25mm]{geometry}

\newcommand{\prob}{\mathbf P}
\newcommand{\ex}{\mathbf E}
\newcommand{\real}{\mathbb R}
\newcommand{\Normal}{\mathcal{N}}

\DeclareMathOperator{\VaR}{VaR}
\DeclareMathOperator{\EVaR}{EVaR}

\DeclareMathOperator{\Pois}{Pois}
\DeclareMathOperator{\NIG}{NIG}
\DeclareMathOperator{\IG}{IG}
\DeclareMathOperator{\Exp}{Exp}
\DeclareMathOperator{\Bern}{Bern}

\numberwithin{equation}{section}
\allowdisplaybreaks

\newtheorem{Theorem}{Theorem}[section]
\newtheorem{Corollary}[Theorem]{Corollary}
\newtheorem{Lemma}[Theorem]{Lemma}
\theoremstyle{definition}
\newtheorem{Definition}[Theorem]{Definition}
\theoremstyle{remark}
\newtheorem{Remark}[Theorem]{Remark}

\begin{document}

\title[Properties of EVaR in relation to selected  distributions]{Properties of the entropic risk measure EVaR\\ in relation to selected  distributions}
\thanks{The first  author is supported by The Swedish Foundation for Strategic Research, grant  no.\ UKR22-0017.
The second author is supported by the Research Council of Finland, decision no.\ 359815.
The first and the second authors acknowledge that the present research is
carried out within the frame and support of the ToppForsk project no.
274410 of the Research Council of Norway with the title STORM: Stochastics
for Time-Space Risk Models.}

\author{Yuliya Mishura$^{1,2}$}
\address{$^1$Taras Shevchenko National University of Kyiv, Ukraine} 
\address{$^2$M\"alardalen University, Sweden}
\email{yuliyamishura@knu.ua}
\author{Kostiantyn Ralchenko$^{1,3}$}
\email{kostiantynralchenko@knu.ua}
\address{$^3$University of Vaasa, Finland}
\author{Petro Zelenko$^4$} 
\email{petro.zelenko@uni-ulm.de}
\address{$^4$Ulm University, Germany}
\author{Volodymyr Zubchenko$^1$}
\email{volodymyr.zubchenko@knu.ua}

\begin{abstract}
Entropic Value-at-Risk (EVaR) measure is a convenient coherent risk measure.   Due to certain difficulties in finding its analytical representation, it was previously calculated explicitly only for the normal distribution. We succeeded  to overcome these difficulties and  to  calculate  Entropic Value-at-Risk (EVaR) measure for Poisson,
compound Poisson, Gamma, Laplace, exponential, chi-squared, inverse Gaussian distribution and  normal inverse Gaussian
distribution with the help of  Lambert function that is a special function, generally speaking,  with two branches.
\end{abstract}



\maketitle

\section{Introduction}
It is well known that financial institutions need to maintain a
certain load on assets to protect against sudden losses associated
with debt default risk, operational risks, market risks, liquidity
etc. To solve this problem, risk measures are used. A risk measure
assigns a real number to a random outcome or risk event that
expresses the degree of risk associated with that random outcome.
This concept has found many applications in various fields such as
finance, actuarial science, operations research and management.

The simplest and best known measure is VaR (Value-at-Risk).
Provided that a stock's fluctuation over time has no peaks, VaR is
useful for predicting the risk associated with a portfolio. But
when we face financial crises of different types, the chance of
VaR to give a complete overview of possible risks decreases, as it
is not sensitive to anything beyond its loss threshold. In
addition, VaR is not coherent. This problem is partially solved by
CVaR (Conditional Value-at-Risk), but not well enough. An
important disadvantage of CVaR is that it cannot be computed
efficiently, even for a sum of arbitrary independent random
variables.
Furthermore, financial markets may face high volatility and instability. In such circumstances, traders and managers use more stringent and conservative measures to manage risk rather than VaR and CVaR. In this case a more complicated measure can be  used, namely a coherent measure of risk called Entropic Value-at-Risk (EVaR), introduced in \cite{Ah12}. The basic properties of this measure are described in the papers   \cite{Ah12,AHMADIJAVID2019225,AP17}.
A quantitative analysis of various risk measures, including EVaR is conducted in \cite{Pich}. Delbaen \cite{FrDel} provides a relation between EVaR and other commonotone risk measures. 
The paper \cite{Pisch} extends and generalizes EVaR by involving R\'enyi entropies. It provides explicit relations among different entropic risk measures, elaborates their dual representations and presents their relations explicitly.
In \cite{AxChow} authors show that EVaR, which is not a dynamic risk
measure in general, can be a finitely-valued dynamic risk measure for at least one value of confidence parameter.
In \cite{Sojal} the concept of EVaR is applied to portfolio selection, and a new mean-EVaR model with uncertain random returns is established.
In \cite{LuxBoy} authors explore portfolio construction in the context of Gaussian mixture returns, aiming to maximize expected exponential utility. Additionally, the authors demonstrate that minimizing EVaR can also be addressed through convex optimization.

Important properties of EVaR are coherence and strong monotonicity in
its domain (see \cite{AHMADIJAVID2019225}), while monotonic risk
measures such as VaR and CVaR lack these properties.
However, EVaR also has some shortcomings.
In particular, there are certain types
of distributions with which this measure cannot be applied (for example, distributions for which the
moment-generating function does not exist).
In addition, the calculation of this measure is often reduced to solving optimization problem,
which leads to large consumption of computation time.
The goal of the present paper is to obtain analytical representations of EVaR for certain risk distribution, which will help to reduce this problem.
Namely, we derive the explicit formulas for EVaR of Poisson, compound Poisson, gamma, Laplace, inverse Gaussian and normal inverse Gaussian distributions, which are widely used in the risk modelling.
It turns out that for first four of these distributions EVaR is expressed trough the so-called Lambert function, which can be efficiently computed with the help of the modern mathematical software. 
In addition, for each distribution we provide graphical illustrations demonstrating the behavior of its EVaR depending on the various distributional parameters.

The paper is organized as follows.
In Section \ref{sec:evar} we recall the definition of EVaR together with its basic properties.
Section \ref{sec:evar-distributions} contains the calculation of EVaR for selected distributions. In subsection \ref{ssec:lambert} we recall the notion of the Lambert function, which is needed for further calculations. 
Subsections \ref{ssec:pois}--\ref{ssec:nig} are devoted to the derivation of EVaR for Poisson, compound Poisson, gamma, Laplace, inverse Gaussian and normal inverse Gaussian distributions respectively. Two appendices supplement the paper: Appendix~\ref{sec:evar-normal} contains the formula for EVaR of the normal distribution and Appendix~\ref{sec:pdfs} collects the probability density functions of selected distributions, considered in Section~\ref{sec:evar-distributions}.
\section{General properties of EVaR}
\label{sec:evar}
    
Let $(\Omega,\mathcal{F},\prob)$ be a probability space and let $\mathcal{L}_0 = \mathcal{L}_0 (\Omega,\mathcal{F},\prob)$ be  space of random variables $X\colon \Omega\rightarrow \real$.
   \begin{Definition}\label{rm}
A \emph{risk measure} $\rho$ is the mapping $\mathcal{L}_0 \rightarrow \overline{\real}$ for which the following conditions are satisfied:
\begin{enumerate}[label=\arabic*)]
\item $\rho(0) = 0$,
\item for all $a \in \real$, $X \in \mathcal{L}_0$ it holds that $\rho(X+a) = \rho(X)+a$,
\item for all  $X,Y \in \mathcal{L}_0$ such that $X \leq Y$  a.s. it holds that $\rho(X)\leq \rho(Y)$.
\end{enumerate}
   \end{Definition}
\begin{Definition}\label{crm}
       Risk measure $\rho\colon \mathcal{L}_0 \rightarrow \overline{\real}$ is called a \emph{coherent risk measure} if the following additional conditions are satisfied:
\begin{enumerate}[label=\arabic*)]
\item $\rho(tX+(1-t)Y) \leq t\rho(X)+(1-t)\rho(Y)$ for all $X,Y \in \mathcal{L}_0$, $t \in [0,1]$,
\item $\rho(tX) = t\rho(X)$ for all $t > 0$ and all $X\in \mathcal{L}_0$.
\end{enumerate}
   \end{Definition}
Property 1) of Definition \ref{crm} is called the convexity of the risk measure, while property 2) symbolizes its homogeneity. 

Among various risk measures, the following
risk measure is often used in modern risk management practice  (\cite{FöllmerSchied+2016,VAR}).
\begin{Definition}
    The \emph{Value-at-Risk} (VaR) with confidence level $\alpha \in (0,1)$ of a random variable $X$ is the smallest number $y$ such that $X$ does not exceed it with a minimum probability of $1-\alpha$, i.e.
    \[\VaR_\alpha(X) := \inf\{x\in \real: F_X(x) >1-\alpha\}.\]
\end{Definition}

However, in this paper we focus on another  risk measure, namely,  EVaR, introduced in \cite{Ah12} according to the
following definition.
\begin{Definition}\label{df:evar}
Let $X \in \mathcal{L}_0$. 
Assume that its moment-generating function $m_X(t) = \ex e^{t X}$ is well defined for all $t\ge0$. An \emph{entropic risk measure} EVaR with a confidence level $\alpha \in [0,1)$ (or a risk level $1-\alpha$) is defined as follows:
    \begin{equation}\label{eq:evar-def}
        \EVaR_\alpha(X):=\inf_{t>0}t^{-1} \log\left(\frac{1}{1-\alpha}m_X(t)\right), \quad \alpha \in [0,1).
    \end{equation}
\end{Definition}
\begin{Remark} It is sufficient  to  assume that moment-generating function $m_X(t)$ is well defined for all $t\in[0,A]$ for some $A>0$.
    
\end{Remark}

According to \cite{Ah12} and \cite{AP17}, EVaR is a coherent risk measure. 
In addition, an important property of EVaR from
\cite{AHMADIJAVID2019225} is the following: 
For $X,Y \in
\mathcal{L}_0$ with the distribution functions $F_X$ and
$F_Y$ respectively,  whose moment-generating functions exist for all $t\in\real$, the following one-to-one correspondence holds:
$\EVaR_\alpha(X) = \EVaR_\alpha(Y)$ for all $\alpha \in [0,1)$ if and only if $F_X(t) = F_Y(t)$ for all $t \in \real$.

One of the shortcomings of EVaR is the difficulty of real-world computations for the models. It was computed mostly for the Gaussian  distribution.  In the next sections we obtain analytical representations which help to overcome indicated shortcoming.

\section{Calculation of EVaR for  selected   distributions}
\label{sec:evar-distributions}
 
In this section we calculate EVaR for   Poisson, compound Poisson, Gamma, inverse Gaussian and normal--inverse Gaussian  distributions.
Note that EVaR for normal distribution is given in \cite{Ah12}, the corresponding formula is given in Appendix. We shall intensively calculate EVaR via so called Lambert function, therefore as the first step, we present its main properties. 

\subsection{Lambert $W$ function}
\label{ssec:lambert}
The Lambert $W$ function \cite{LF} is a multi-valued inverse of the function $x \mapsto xe^x$. 
In other words, $W$ is defined as a function satisfying
\begin{equation}\label{eq:w-equation}
W(x) e^{W(x)} = x,\quad x\in \real.
\end{equation}
For positive $x$ this function is single-valued, for $x < -\frac{1}{e}$ there is no inverse function. For $-\frac{1}{e}\leq x \leq 0$ there are two possible real values of $W(x)$ (see Fig.~\ref{W_function}). We denote the branch satisfying $W(x)\geq -1$ by $W_0(x)$ and the branch satisfying $ W(x)\leq-1$ by  $W_{-1}(x)$ \cite{LF}.
Note that 
\begin{itemize}
\item $W_0(-\frac{1}{e}) = W_{-1}(-\frac{1}{e}) = -1$;
\item $W_0(x)$ is defined for all $x\ge-\frac{1}{e}$ and is strictly increasing;
\item $W_{-1}(x)$ is defined for $-\frac{1}{e} \le x <0$ and is strictly decreasing.
\end{itemize}

\begin{figure}
    \centering
        \includegraphics[scale=0.5]{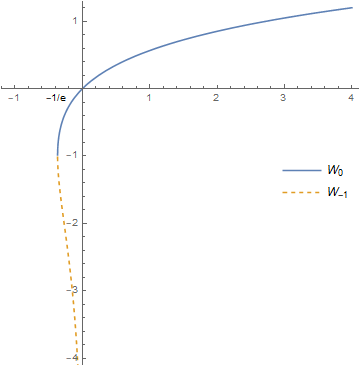}
        \caption{
            The two real branches of the Lambert $W$ function: $W_0(x)$ (solid  line) and $W_{-1}(x)$ (dashed  line).}
        \label{W_function}
\end{figure}
$W_0(x)$ is referred to as the \emph{principal branch} of the W function. It is obvious that $W_0(x)$ preserves the sign of $x.$

The derivative of $W$ equals (see \cite[Eq.~(3.2)]{LF})
\begin{equation} \label{eq:W-deriv}
W'(x) = \frac{1}{(1+W(x)) \exp(W(x))}
= \frac{W(x)}{x(1+W(x))}, 
\quad\text{if } x \ne 0.
\end{equation}

Taking logarithms of both sides of \eqref{eq:w-equation} and rearranging terms we get the following transformations:
\begin{gather}
    \log(W_0(z)) = \log(z) - W_0(z), \quad z>0,\notag
    \\
    \label{ln}
    \log(-W_k(z)) = \log(-z) - W_k(z),\quad k \in\{-1,0\}, \quad z\in[-\tfrac1e,0).
\end{gather}
A more general formula, involving branches of complex logarithm, can be found in \cite{LF2}.

\subsection{Poisson distribution}
\label{ssec:pois}

Let us calculate EVaR for a Poisson distribution. Denote $\beta = -\log (1-\alpha) - \lambda$.
\begin{Theorem}\label{th:evar-pois}
    Let $X \sim \Pois(\lambda)$, $\lambda>0$, then for all $\alpha \in [0,1)$
    \begin{align}\label{eq:evar-pois}
        \EVaR_\alpha(X) &=
        \begin{cases}
            \frac{\beta}{W_0(\frac{\beta}{e\lambda})}, & \beta \neq 0,\\
            e\lambda, &\beta=0,
        \end{cases}\\
        &= e\lambda e^{W_0(\frac{\beta}{e\lambda})},
        \label{eq:evar-pois-2}
    \end{align}
    where $W_0$ stands for the principal branch of the Lambert function.  
    The expression for $\EVaR_\alpha(X)$ is continuous in $\alpha$, and for any fixed $\alpha$  is continuous in $\beta\in \real$.
\end{Theorem}
\begin{proof}
    Recall that the moment-generating function of the Poisson distribution equals:
    \begin{equation}
        m_X(t) = \exp\{\lambda(e^t-1)\}, \quad t\in\real.\nonumber
    \end{equation}
    Therefore, by Definition \ref{df:evar},
    \begin{equation}
        \EVaR_\alpha(X) = \inf_{t > 0}\frac{1}{t}\left(-\log(1-\alpha) + \lambda(e^t-1)\right) = \inf_{t > 0}\frac{\beta + \lambda e^t}{t}.\nonumber
    \end{equation}
    Let us find the the infimum of the function $f(t) = \frac{\beta + \lambda e^t}{t}$ over $t>0$. Note that $f$ is continuous on $(0, +\infty)$ and tends to $+\infty$ at $0$ and $+\infty$. It is obvious at $+\infty$ but at $0$ it follows from the fact that $$\beta+\lambda=-\log(1-\alpha)>0.$$ Therefore the infimum is  in fact minimal value and achieved inside $(0, +\infty)$.
    
Let $\beta\neq 0$. Then the derivative of the function $f$ equals
    \[
    f'(t) = \frac{\lambda e^t(t-1)-\beta}{t^2}.
    \]
    The condition $f'(t) = 0$ for some $t>0$ leads to the equation
    \begin{equation}\label{eq:pois-minimization}
     e^{t-1}(t-1)=\frac{\beta}{e\lambda}.
    \end{equation}
Since for all $\alpha\in[0,1)$
\[
\frac{\beta}{e\lambda} = \frac{1}{e\lambda}\log\frac{1}{1-\alpha} - \frac{1}{e} \ge - \frac{1}{e},
\]
we see that the solution of \eqref{eq:pois-minimization} can be expressed through the principal branch $W_0$ of Lambert function as follows
    \begin{equation*}
    t^*= 1+W_0\left(\frac{\beta}{e\lambda}\right).
    \end{equation*}
 The choice of the the principal branch $W_0$ of  Lambert function corresponds to the fact that we consider  $t>0$. 

Let us check that the sufficient condition for the local minimum at point $t^*$ is fulfilled.
Using the equality $e^{W(x)} = x / W(x)$ (see \eqref{eq:w-equation}), we get
    \begin{align*}
        f''(t^*) &= \frac{2\beta + \lambda e^t(t^2- 2t+2)}{t^3}\bigg|_{t= t^*} 
        =\frac{2\beta + \lambda e^{1+W_0(\frac{\beta}{e\lambda})}\left(1+\left(W_0\left(\frac{\beta}{e\lambda}\right)\right)^2\right)}{\left(1+W_0\left(\frac{\beta}{e\lambda}\right)\right)^3}
        \\
        &=\frac{\beta\left(W_0\left(\frac{\beta}{e\lambda}\right)+1\right)^2}{\left(1+W_0\left(\frac{\beta}{e\lambda}\right)\right)^3W_0\left(\frac{\beta}{e\lambda}\right)}
        =\frac{\beta}{\left(1+W_0\left(\frac{\beta}{e\lambda}\right)\right) W_0\left(\frac{\beta}{e\lambda}\right)}
       > 0, 
    \end{align*}
    because $W_0$
preserves the sign of argument and therefore $\frac{\beta}{W_0\left(\frac{\beta}{e\lambda}\right)}>0$  and  $\left(1+W_0\left(\frac{\beta}{e\lambda}\right)\right)=t^*>0$ also. So, the point $t=t^*$ minimizes $f(t)$ over $t>0$ and the minimum equals 
    $f(t^*) = \frac{\beta}{W_0(\frac{\beta}{e\lambda})}$.
    
In the case $\beta = 0$ we get that the minimum is achieved at point $t=1$ and equals
    \begin{equation*}
        \EVaR_\alpha(X) = \inf_{t > 0}\frac{\lambda e^t}{t} = \lambda e.
    \end{equation*}
Thus the formula \eqref{eq:evar-pois} is proved.
Combining it with \eqref{eq:w-equation}, we derive the representation \eqref{eq:evar-pois-2}, which implies the continuity.
\end{proof}

\begin{figure}[t]
    \centering
    \includegraphics[scale=0.55]{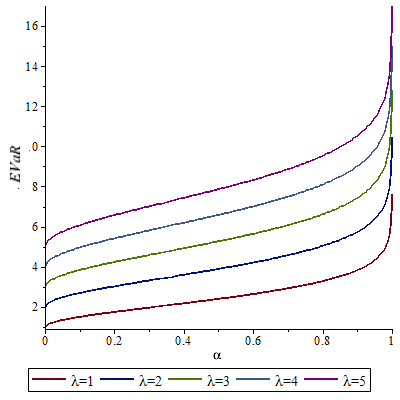}
    \includegraphics[scale=0.3]{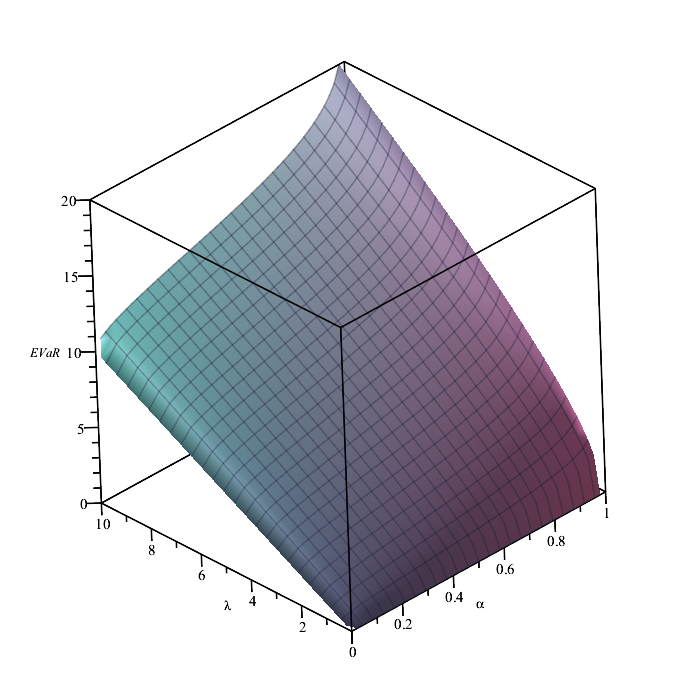}
    \caption{
    EVaR for Poisson distribution with constant intensity $\lambda$.}
\end{figure}

\subsection{Compound Poisson distribution}
Risk measures are often used for models with a large number of
identical losses. This could be, for example, losses on an
insurance portfolio. In such models, the losses are usually
modelled as independent and identically distributed random
variables, and the loss intensity will have a discrete
distribution, such as binomial, Poisson, or negative binomial.
Suppose that a collective risk model $X:= \sum_{i=1}^{\eta}\xi_i$
(here $\sum_{i=1}^0 = 0$ by convention)
is given, where the intensity of insurance cases is modelled by
the Poisson distribution $\eta \sim \Pois(\lambda)$, and the losses
$\{\xi_i,i\ge1\}$ are independent, identically distributed
random variables with the moment-generating function $m_\xi(t) =
\ex e^{t\xi}$. 
It is well known that the moment-generating function of the compound Poisson distribution is given by
\begin{equation}\label{eq:mgf-comppois}
    m_X(t)= \exp\{\lambda (m_\xi(t)-1)\}.
\end{equation}
Then for all $\alpha \in (0,1)$,
\begin{equation}\label{eq:evar-comppois}
    \EVaR_\alpha(X) = \inf_{t>0}\frac{-\log(1-\alpha)+\lambda(m_\xi(t)-1)}{t}.
\end{equation}

Now let us consider Bernoulli distribution for the values of
losses.
\begin{Theorem}
    Let $\eta \sim \Pois(\lambda)$, $\xi_i \sim \Bern(p)$, $i\ge1$, be independent, identically distributed random variables, and
    $X:= \sum_{i=1}^{\eta}\xi_i$. Then for any $\alpha \in (0,1)$
    \[
        \EVaR_\alpha(X) =
        \begin{cases}
            \frac{\beta}{W_0(\frac{\beta}{e\lambda p})}, & \beta \neq 0,\\
            e\lambda p, &\beta=0,
        \end{cases}
    \]
    where $\beta = -\log(1-\alpha) - \lambda p$, and $W_0$ stands for the principal branch of the Lambert function.   
    The expression for $\EVaR_\alpha(X)$ is continuous in $\alpha$.
\end{Theorem}
\begin{proof}
By inserting the moment-generating function $m_\xi(t) = 1-p +pe^t$ of the Bernoulli distribution
into \eqref{eq:mgf-comppois} it is not hard to see that $X\sim\Pois(\lambda p)$. Hence, the result follows from Theorem \ref{th:evar-pois}.
 \end{proof}

Now let us consider the case where losses have a
centered normal distribution.
\begin{Theorem}
    Let $X:= \sum_{i=1}^{\eta}\xi_i$, where $\eta \sim \Pois(\lambda)$,  and $\forall \geq 1: \xi_i \sim \Normal(0,\sigma^2)$ are independent, identically distributed random variables. Then for all $\alpha \in (0,1)$,
    \begin{equation}\label{eq:comp-pois-normal}
        \EVaR_\alpha(X) =
        \begin{cases}
           \beta\sigma\frac{\left(2W_0(\gamma)+1\right)^{1/2}}{2W_0(\gamma)}, & \beta \neq 0,\\
            \lambda\sigma\sqrt{e}, & \beta = 0.
        \end{cases}
    \end{equation}
    where $\beta = -\log(1-\alpha) - \lambda$, $\gamma = \frac{\beta}{2\sqrt{e}\lambda}$.
    In addition, $\EVaR_\alpha(X)$ is continuous in $\alpha$.
\end{Theorem}
\begin{proof}    
Substituting the moment-generating function of the normal distribution $m_\xi(t) =\linebreak \exp\{\frac{1}{2}t^2\sigma^2\}$ into \eqref{eq:evar-comppois}, we see that
    \begin{equation}
        \EVaR_{\alpha} = \inf_{t>0}\frac{-\log(1-\alpha)+\lambda(m_\xi(t)-1)}{t} = \inf_{t > 0}\frac{\lambda \exp\left\{\frac{\sigma^2t^2}{2}\right\}+\beta}{t},\nonumber
    \end{equation}
So we need to minimize the function
    \[
    f(t) = \frac{\lambda \exp\left\{\frac{\sigma^2t^2}{2}\right\}+\beta}{t}.
    \]
     First, let us consider the case $\beta \neq 0$.
     The derivative of $f$ equals
    \begin{equation}
        f'(t) = \frac{\lambda \exp\left\{\frac{\sigma^2t^2}{2}\right\}(\sigma^2t^2-1)-\beta}{t^2}. \nonumber
    \end{equation}
    Then the condition $f'(t) = 0$ implies the equation
    \begin{equation}\label{eq:equation-cpois}
        \exp\left\{\frac{\sigma^2t^2-1}{2}\right\}\cdot \frac{\sigma^2t^2-1}{2} =\frac{\beta}{2\sqrt{e}\lambda}\eqqcolon\gamma,
    \end{equation}
    which can be solved with the help of the Lambert function (see subsection \ref{ssec:lambert}).
    To this end, we need to check that the right-hand side of equation \eqref{eq:equation-cpois} exceeds $-\frac 1e$.
    This condition holds, because
    \begin{equation}\label{eq:inequalities-cpois}
       \gamma = \frac{\beta}{2\sqrt{e}\lambda} 
        = \frac{-\log(1-\alpha) - \lambda}{2\sqrt{e}\lambda} 
        \ge -\frac{1}{2\sqrt{e}}
        > -\frac{1}{e} 
    \end{equation}
    Therefore, by the definition \eqref{eq:w-equation} of the Lambert function, the solution of \eqref{eq:equation-cpois} is given by
       \begin{equation}
        t_*:= \frac{1}{\sigma}\left(1+2W_0(\gamma)\right)^{\frac12}.\nonumber
    \end{equation}
    Note that we choose the principal branch $W_0$ of the Lambert function, because for $W_{-1}$ the expression under square root is negative. For $W_0$ it is always positive. Indeed, by \eqref{eq:inequalities-cpois},
    $\gamma\ge -\frac{1}{2\sqrt{e}}$. Due to monotonicity of $W_0$, we have
    \[
    W_0(\gamma)\ge W_0\left(-\frac{1}{2\sqrt{e}}\right) = -\frac12,
    \]
    where the last equality follows from the definition of the Lambert function (it easy to see from  \eqref{eq:w-equation} that $W_0 (x) = -\frac12$ for $x=-\frac{1}{2\sqrt{e}}$).
    Thus the value $t_*$ is well defined for any $\alpha\in(0,1)$.
   Let us check the sufficient condition for a local minimum.
   The second derivative of $f$ at $t=t_*$ equals
    \[
        f''(t_*) = \frac{2\beta + \lambda \exp\left\{\frac{\sigma^2t_*^2}{2}\right\}\left(2+\sigma^4t_*^4-\sigma^2t_*^2\right)}{t_*^3}
    \]
    Since $t_*$ satisfies the equation \eqref{eq:equation-cpois}, we see that
    \begin{equation}\label{eq:equation-cpois1}
    \lambda \exp\left\{\frac{\sigma^2t_*^2}{2}\right\}
    = \frac{\beta}{\sigma^2t_*^2-1}
    = \frac{\beta}{2W_0(\gamma)}.
    \end{equation}
    Therefore 
    \begin{align*}
        f''(t_*)&=\frac{2\beta + \frac{\beta}{2W_0(\gamma)}\left(2+\left(1+2W_0(\gamma)\right)^2- \left(1+2W_0(\gamma)\right)\right)}{\sigma^{-3}\left(2W_0(\gamma)+1\right)^{3/2}} 
    \\
        &= \frac{\sigma^{3}\beta\left(2\left(W_0(\gamma)\right)^2+W_0(\gamma) + 3\right)}{\left(1+2W_0(\gamma)\right)^{3/2}W_0(\gamma)}.
    \end{align*}
   Note that for any $\beta\ne0$
    \[
        \frac{\beta}{W_0(\gamma)} = \frac{\beta}{W_0\left(\frac{\beta}{2\sqrt{e}\lambda}\right)} > 0
    \]
    and
    \[
        2\left(W_0(\gamma)\right)^2+W_0(\gamma) + 3 > 0,
    \]
    because $2x^2+x+3 > 0$ for any $x\in\real$.
    As a result, $f''(t) > 0$. The sufficient condition for a local minimum is satisfied. 
    Therefore, the minimal value of $f(t)$ is attained at $t=t_*$ and equals
    \[
    f(t_*) = \frac{\lambda \exp\left\{\frac{\sigma^2t_*^2}{2}\right\}+\beta}{t_*}
    =  \frac{\frac{\beta}{2W_0(\gamma)} + \beta}{\sigma^{-1}(1+2W_0(\gamma))^{1/2}}
    = \beta\sigma\frac{(2W_0(\gamma)+1)^{1/2}}{2W_0(\gamma)},
    \]
    where we have used \eqref{eq:equation-cpois1}.
    
    Let us consider the case $\beta = 0$. The problem reduces to the finding the value of
    \begin{equation}
        \inf_{t>0}\frac{\lambda \exp\left\{\frac{\sigma^2t^2}{2}\right\}}{t}.\nonumber
    \end{equation}
    The minimum is achieved at the point $t=\frac{1}{\sigma}$. So $\EVaR_\alpha(X) = \lambda \sigma \sqrt{e}$. 
    
    Thus the formula \eqref{eq:comp-pois-normal} is proved. It remains to verify the continuity of $\EVaR_\alpha(X)$ at the point $\beta=0$.
    Let us check the convergence of the previous expression by applying L'H\^opital's rule since the numerator and denominator tend to 0. 
    Their derivatives can be computed using \eqref{eq:W-deriv} as follows
    \[
        \frac{d}{d\beta} W_0\left(\frac{\beta}{2\sqrt{e}\lambda}\right) = \frac{1}{2\sqrt{e}\lambda \exp\{W_0(\gamma)\}(1+W_0(\gamma))}.
    \]
    \begin{multline*}
        \frac{d}{d\beta} \left(\beta\sigma\left(2W_0\left(\frac{\beta}{2\sqrt{e}\lambda}\right)+1\right)^{\frac12}\right) \\*
        = \sigma(2W_0(\gamma)+1)^\frac12 + \frac{\beta\sigma}{2\sqrt{e}\lambda \exp\{W_0(\gamma)\}(1+W_0(\gamma))(2W_0(\gamma) + 1)^{1/2}}.\nonumber
    \end{multline*}
   Therefore, the limit equals
    \begin{align*}
        \MoveEqLeft[1]
        \lim_{\beta\rightarrow0}\EVaR_\alpha(X) \\
        &=\lim_{\beta\rightarrow 0}\Biggl[\sqrt{e}\lambda \exp\{W_0(\gamma)\}(1+W_0(\gamma))\\*
        &\qquad\qquad\times\Biggl(\sigma(2W_0(\gamma)+1)^\frac12 
        \\
        &\qquad\qquad\qquad+\frac{\beta\sigma}{2\sqrt{e}\lambda \exp\{W_0(\gamma)\}(1+W_0(\gamma))(2W_0(\gamma) + 1)^{1/2}}\Biggr)\Biggr] \\
        &= \sqrt{e}\lambda\sigma.\nonumber
    \end{align*}
    Thus, the continuity is proved.
\end{proof}

\begin{figure}[t]
    \centering
    \includegraphics[scale=0.43]{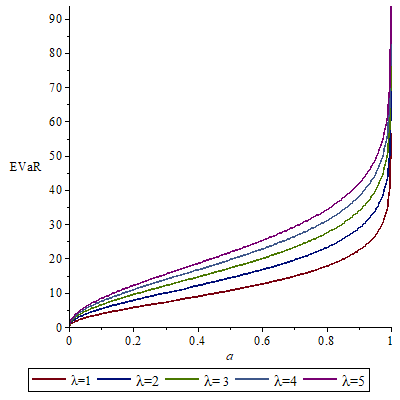}
    \includegraphics[scale=0.3]{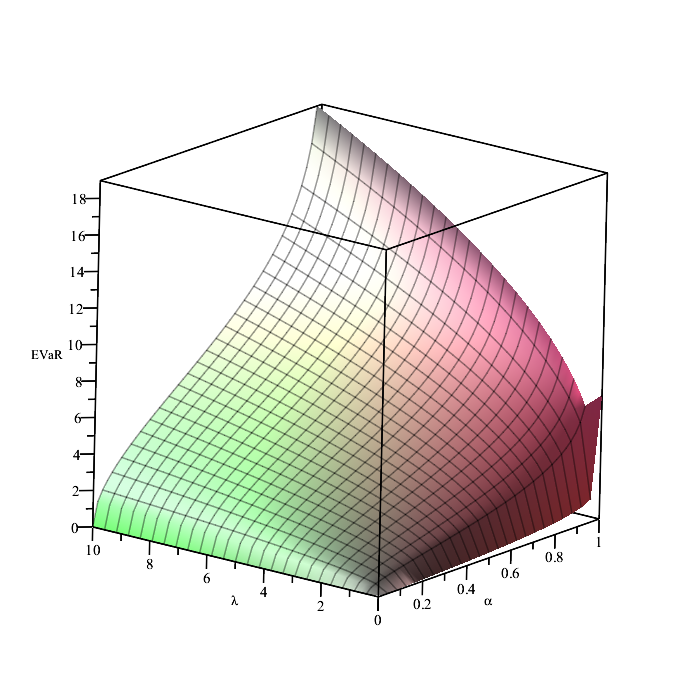}
    \caption{
        EVaR for compound Poisson distribution with normal distribution of jumps ($\sigma=1$).}
\end{figure}

\subsection{Gamma distribution}
\begin{Theorem}
    Let $X\sim G(k,\theta)$ have a gamma distribution. Then for any $\alpha \in (0,1)$,
    \begin{equation}\label{eq:evar-gamma}
        \EVaR_\alpha(X) = - k\theta W_{-1}\left(-e^{-1}(1-\alpha)^{\frac{1}{k}}\right).
    \end{equation}
    Here $W_{-1}$ denotes the branch of Lambert function that does not exceed $-1$.
    The expression for $\EVaR_\alpha(X)$ is continuous in $\alpha$.
\end{Theorem}
\begin{proof}
    The moment-generating function of the gamma distribution is given by
    \begin{equation}
        m_X(t) = (1-\theta t)^{-k}, \quad t < \frac{1}{\theta}.\nonumber
    \end{equation}
    Then
    \begin{align*}
        \EVaR_\alpha(X) &= \inf_{t \in (0,\theta^{-1})}\frac{1}{t}\Bigl(-\log(1-\alpha) - k \log(1-\theta t)\Bigr)\\
        &= \inf_{t \in (0,\theta^{-1})}\frac{b -k\log(\frac{1}{\theta}-t)}{t},\nonumber
    \end{align*}
    where $b: = -\log(1-\alpha) - k \log\theta$.

    Let us denote $f(t) := \frac{1}{t} (b -k\log(\frac{1}{\theta}-t))$.
    We are looking for the point of local minimum of the function $f$ on $(0,\frac1\theta)$. Its derivative equals
    \[
        f'(t) = \frac{\frac{\theta k t}{1-\theta t} + k \log(\frac{1}{\theta} -t) - b}{t^2}.
    \]
    The condition $f'(t) = 0$ leads to the equation
    \[
     k \log\left(\frac{1}{\theta} -t\right)=b-\frac{\theta k t}{1-\theta t},
    \]
    which is equivalent to
    \[
        \frac{1-\theta t}{\theta} = \exp\left\{\frac{b}{k}+1\right\} \exp\left\{-\frac{1}{1-\theta t}\right\},
    \]
    or
    \begin{equation}\label{eq:equation-gamma}
        -\frac{1}{1-\theta t}\exp\left\{-\frac{1}{1-\theta t}\right\} = -\frac{1}{\theta} \exp\left\{-\frac{b}{k}-1\right\}.
    \end{equation}
    The equation \eqref{eq:equation-gamma} can be solved using the Lambert function, since its right-hand side is not less than $-\frac{1}{e}$:
    \[
        -\frac{1}{\theta}\exp\left\{-\frac{b}{k}-1\right\} = -e^{-1}(1-\alpha)^{\frac{1}{k}} \geq -e^{-1}.
    \]
    We see from the condition $t<\frac1\theta$ that $-\frac{1}{1-\theta t}\leq -1$, hence, the branch $W_{-1}$ of the Lambert function should be chosen. We arrive at the equation
    \[
        -\frac{1}{1-\theta t} = W_{-1}\left(-e^{-1}(1-\alpha)^{\frac{1}{k}} \right),
    \]
    its solution is
    \[
        t_* := \theta^{-1} + \frac{1}{\theta W_{-1}(-e^{-1}(1-\alpha)^{\frac{1}{k}})}\nonumber
    \]
    Let us check the sufficient conditions for a local minimum. The second derivative equals
    \begin{equation}\label{eq:secdergamma}
        f''(t_*) = \frac{1}{t_*^3}\left(2\log\frac{1}{1-\alpha} + \frac{\theta kt_*(3\theta t_*-2)}{(1-\theta t_*)^2} - 2k\log(1-\theta t_*)\right).
    \end{equation}
    Let $z:=-e^{-1}(1-\alpha)^{\frac{1}{k}}$.
    Substituting the value of $t_*$ into \eqref{eq:secdergamma} and applying \eqref{ln}, we get
        \begin{align*}
          f''(t_*) & = \frac{\theta^3(W_{-1}(z))^3}{(W_{-1}(z)+1)^3}\biggl(2\log\left(\frac{1}{1-\alpha}\right) + k(W_{-1}(z))^2 \\*
          &\qquad\qquad\qquad\qquad\quad
          + 4W_{-1}(z) + 2\log(-W_{-1}(z))\biggr) \\
          &=\frac{\theta^3(W_{-1}(z))^3k(W_{-1}(z)+1)^2}{(W_{-1}(z)+1)^3} \geq 0,
        \end{align*}
        because $W_{-1}(z) \le -1$.
        Thus we have proved that the minimum of $f$ is achieved at the point $t=t_*$. After substituting $t_*$ into $f(t)$ and simplifying the resulting expression, we arrive at \eqref{eq:evar-gamma}.
        The continuity of $\EVaR_\alpha(X)$ follows from the continuity of $W_{-1}$.
    \end{proof}
\begin{figure}[t]
    \centering
    \includegraphics[scale=0.32]{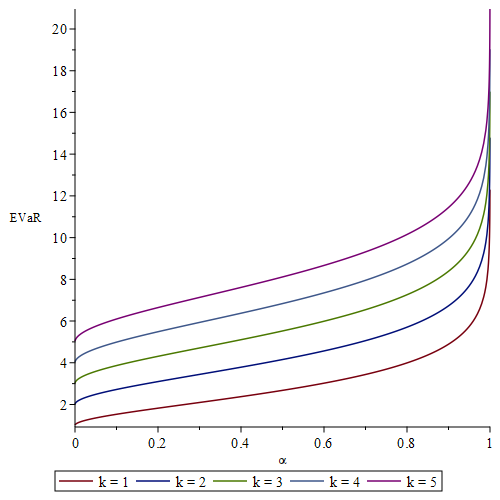}
    \includegraphics[scale=0.32]{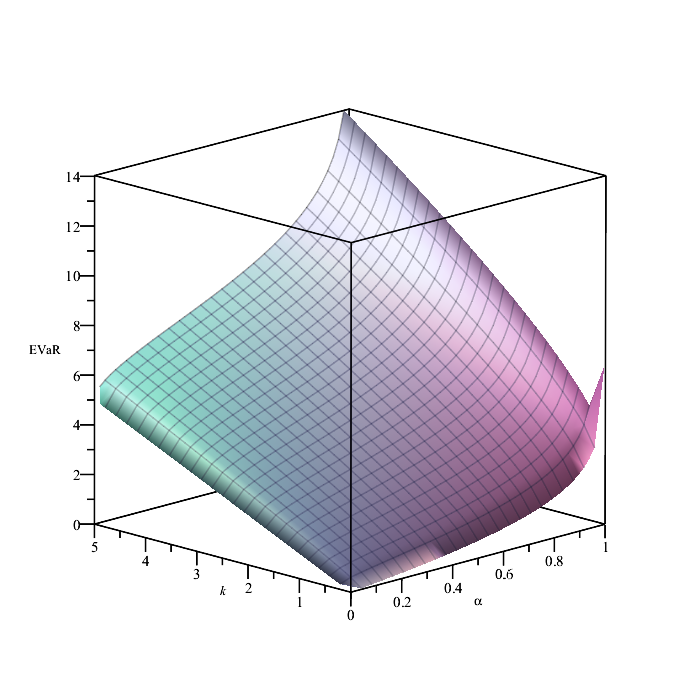}
    \caption{
        EVaR for Gamma distribution with parameter $\theta=1$.}
\end{figure}
\begin{Corollary}
    Let $X \sim \Exp(\lambda)$. Since $\Exp(\lambda) = G(1,\frac{1}{\lambda})$, we see that for all $\alpha \in (0,1)$,
    \begin{equation}
        \EVaR_\alpha(X) = - \frac{1}{\lambda} W_{-1}\left(-\frac{1}{e}(1-\alpha)\right).\nonumber
    \end{equation}
\end{Corollary}
\begin{Corollary}
    Suppose $X \sim \chi^2(k)$ has a chi-squared distribution. Since $X \sim G(\frac{k}{2}, 2)$, we see that for all $\alpha \in
    (0,1)$,
    \begin{equation}
        \EVaR_\alpha(X) = - kW_{-1}\left(-\frac{1}{e}(1-\alpha)^{\frac{2}{k}}\right).\nonumber
    \end{equation}
\end{Corollary}

\subsection{Laplace distribution}
\begin{Theorem}
    Let $X \sim L(\mu, b)$ have a Laplace distribution. Then for all $\alpha \in (0,1)$,
    \begin{equation}\label{eq:evar-lapl}
        \EVaR_\alpha(X) = \mu - bW_{-1}(\gamma)\left(1+\frac{2}{W_{-1}(\gamma)}\right)^{\frac12},
    \end{equation}
    where $\gamma=-2e^{-2}(1-\alpha)$.
    The expression for $\EVaR_\alpha(X)$ is continuous in $\alpha$.
\end{Theorem}
\begin{proof}
    The moment-generating function equals
    \begin{equation}
        m_X(t) = \frac{e^{\mu t}}{1-b^2t^2}, \quad |t|<\frac{1}{b}.\nonumber
    \end{equation}
    Then
    \begin{align}
        \EVaR_\alpha(X)  =  \inf_{t \in (0,b^{-1})}\frac{1}{t}\left(a +\mu t - \log\left(1-b^2t^2\right)\right),\nonumber
    \end{align}
    where $a := -\log(1-\alpha)$. Let us find the minimum of the function
    \begin{equation}\label{eq:laplace-obj}
    f(t) = \frac{1}{t}\left(a +\mu t - \log\left(1-b^2t^2\right)\right).
    \end{equation}
    Its first derivative equals
    \[
        f'(t) = - \frac{a(b^2t^2-1) - b^2t^2(\log(1-b^2t^2)-2) + \log(1-b^2t^2)}{t^2(b^2t^2-1)}.
    \]
    From the condition $f'(t) = 0$, we get the equation
    \[
    a+2 - \frac{2}{1-b^2t^2} = \log(1-b^2t^2)
    \]       
        or
    \begin{equation}\label{eq:laplace-eq}
    -\frac{2}{1-b^2t^2} \exp\left\{-\frac{2}{1-b^2t^2}\right\} = -2e^{-a-2}
    \eqqcolon\gamma.
    \end{equation}
This equation can be solved via the Lambert function, because its right-hand side $\gamma = -2e^{-2}(1-\alpha) \geq -e^{-1}$.
Since $-\frac{2}{1-b^2t^2} < -1$, we see that the branch $W_{-1}$ should be chosen. Thus, we arrive at the following solution to \eqref{eq:laplace-eq}:
    \begin{equation}
        t_*= \frac{1}{b}\left(1+\frac{2}{W_{-1}(\gamma)}\right)^{\frac12}.\nonumber
    \end{equation}
Note that $W_{-1}(t)$ is a decreasing function at $(-e^{-1},0)$, therefore
    \begin{equation}
        W_{-1}(\gamma) = W_{-1}(-2e^{-2}(1-\alpha)) < W_{-1}(-2e^{-2}) = -2,\nonumber
    \end{equation}
so the expression under square root is positive and $t_*$ is well defined.

The second derivative of $f$ equals
    \begin{align*}
        f''(t_*) &= \frac{2}{t_*^3}\left(a+\frac{b^2t_*^2(3b^2t_*^2-1)}{(1-b^2t_*^2)^2} - \log(1-b^2t_*^2)\right) 
    \\
        &= \frac{1}{t_*^3} \left(2a+(W_{-1}(\gamma))^2 + 5W_{-1}(\gamma)+ 6 - 2\log 2 + 2\log(-W_{-1}(\gamma))\right)
    \\
        &= \frac{1}{t_*^3} \left((W_{-1}(z))^2 + 3W_{-1}(z)+ 2\right) > 0,
    \end{align*}
    where we have used the relation \eqref{ln}.
    Consequently, $t_*$ is a minimum point.   Substituting  $t_*$ into the function \eqref{eq:laplace-obj} and simplifying the resulting expression,  we   get \eqref{eq:evar-lapl}.
\end{proof}
\begin{figure}[t]
    \centering
    \includegraphics[scale=0.5]{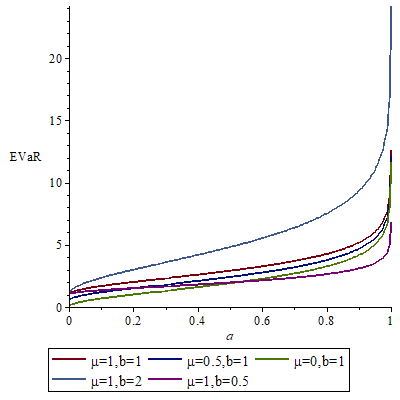}
    \caption{
        EVaR for Laplace distribution.}
\end{figure}

\subsection{Inverse Gaussian distribution}

\begin{Theorem}
    Let $X \sim \IG(\mu,\lambda)$ have an inverse Gaussian distribution, or so-called Wald distribution. Then
    \[
    \EVaR_\alpha(X) = \mu\left(\delta+\sqrt{\delta^2-1}\right),
    \]
    where $\delta = 1-\frac{\mu}{\lambda}\log(1-\alpha)$.
\end{Theorem}
\begin{proof}
    The corresponding moment-generating function equals:
    \[
        m_X(t) = \exp\left\{\frac{\lambda}{\mu}\left(1-\left(1-\frac{2\mu^2t}{\lambda}\right)^{\frac12}\right)\right\},
        \quad t< \frac{\lambda}{2\mu^2}.
    \]
    Substituting it into the definition of $\EVaR_\alpha$, we get the following optimization problem:
    \begin{equation}
        \EVaR_{\alpha}(X) = \inf_{t > 0}\frac{1}{t}\left(-\log(1-\alpha) + \frac{\lambda}{\mu}\left(1-\left(1-\frac{2\mu^2t}{\lambda}\right)^{\frac12}\right)\right).\nonumber
    \end{equation}
  Let denote $c:=-\log(1-\alpha)$.
    Let us find the minimum point of the objective function
    \[
    f(t) = \frac{1}{t}\left(c + \frac{\lambda}{\mu}-\frac{\lambda}{\mu}\left(1-\frac{2\mu^2t}{\lambda}\right)^{\frac12}\right)
    = \frac{1}{t}\left(c + \frac{\lambda}{\mu}-\frac{\lambda}{\mu}z(t)\right),
    \]
    where $z(t):=(1-\frac{2\mu^2t}{\lambda})^{1/2}$.
    Then the derivative of $f(t)$ equals
    \begin{equation}\label{eq:wald-der}
    f'(t) = \frac{1}{t^2} \left(-\frac{\lambda}{\mu}z'(t)t - c -\frac{\lambda}{\mu} + \frac{\lambda}{\mu} z(t)\right)
    \end{equation}
    where
    \[
    z'(t) = - \frac{\mu^2}{\lambda}\left(1-\frac{2\mu^2t}{\lambda}\right)^{-1/2}
    = - \frac{\mu^2}{\lambda z(t)}
    \]
    Using this formula together with the inverse relation 
    $t = \frac{\lambda}{2\mu^2}(1-z^2(t))$, we can transform \eqref{eq:wald-der} as follows:
    \begin{align*}
        f'(t) &= \frac{1}{t^2} \left(\frac{\lambda}{2 \mu z(t)}\left(1-z^2(t)\right) - c -\frac{\lambda}{\mu} + \frac{\lambda}{\mu} z(t)\right)
        \\
        &= \frac{\lambda}{2\mu t^2 z(t)}\left(z^2(t) - 2\delta z(t) +1\right),
    \end{align*}
    where $\delta = \frac{c\mu}{\lambda}+1 \ge 1$.
    We see that the condition $f'(t) = 0$ is satisfied if and only if
    $z(t) = \delta \pm \sqrt{\delta^2 - 1}$.
    The root with ``plus'' sign is bigger than 1, which is not possible for $t>0$.
    Therefore, we need to consider only the value
    \[
    z_* = z(t_*) = \delta - \sqrt{\delta^2 - 1},
    \]
    which corresponds to
    \[
    t_* = \frac{\lambda}{2\mu^2}\left(1-\left( \delta - \sqrt{\delta^2 - 1}\right)^2\right)
    \in \left(0,\frac{\lambda}{2\mu^2}\right).
    \]
    Note that $z(t)$ is a decreasing function and
    \begin{itemize}
        \item $f'(t) < 0$ for all $z>z_*$ (i.e.\ for all $0< t < t_*$),
        \item $f'(t) > 0$ for all $z<z_*$ (i.e.\ for all $t_*< t < \frac{\lambda}{2\mu^2}$).
    \end{itemize}
   This means that the minimal value of $f(t)$ is achieved at $t=t_*$.
   It equals
   \begin{align*}
   f(t_*) &= \frac{1}{t_*}\left(c + \frac{\lambda}{\mu}-\frac{\lambda}{\mu}z_*\right)
   = \frac{\lambda}{\mu}\cdot\frac{\delta-z_*}{t_*}
   = \frac{2\mu\sqrt{\delta^2 - 1}}{1-\left( \delta - \sqrt{\delta^2 - 1}\right)^2}
   \\
   &= \frac{\mu}{\delta -\sqrt{\delta^2 - 1} }
   = \mu\left(\delta +\sqrt{\delta^2 - 1} \right).\qedhere
   \end{align*}
\end{proof}
\begin{figure}[t]
\centering
    \includegraphics[scale=0.36]{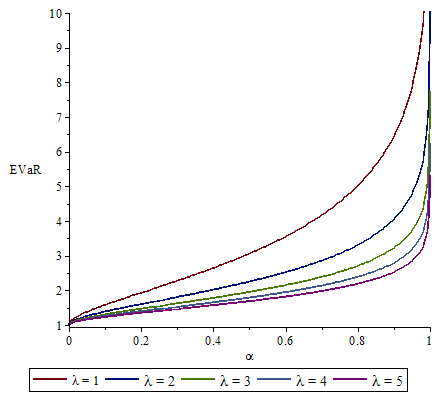}
    \includegraphics[scale=0.36]{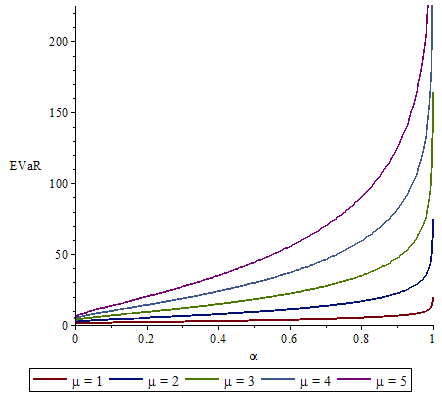}\\
    \includegraphics[scale=0.4]{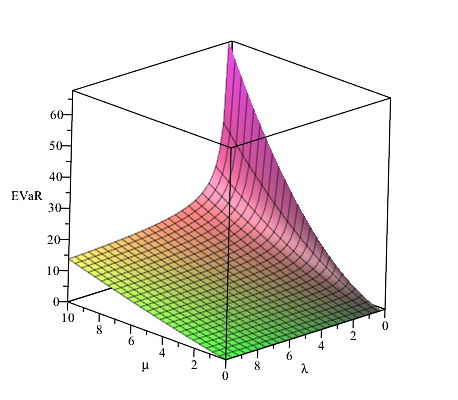}
    \includegraphics[scale=0.4]{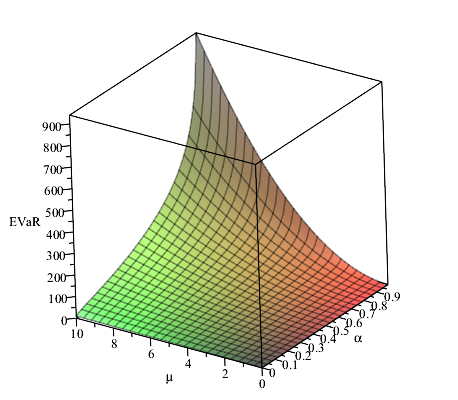}
    \caption{
        EVaR for Wald distribution. The first figure is for fixed $\mu=1$, the second figure is for fixed $\lambda=1$, third for fixed $\alpha=0.05$ and the last one for fixed $\lambda=1$.}
\end{figure}

\subsection{Normal inverse Gaussian distribution}
\label{ssec:nig}
Consider the class of normal inverse Gaussian distributions $\NIG(\alpha,
\beta, \mu, \delta)$. This is a flexible system of distributions
that includes distributions with heavy tails and skewed
distributions.
\begin{Theorem}
    Let $X \sim \NIG(\alpha, \beta, \mu, \delta)$ be a normal inverse Gaussian distribution, $0\le |\beta| < \alpha$, $\mu\in\real$, $\delta>0$.
    Then for the level $\alpha' \in [0,1)$:
    \begin{equation}\label{eq:evar-nig}
    \EVaR_{\alpha'}(X) = \mu + \frac{\delta}{t_*}
    \left(\varphi - \sqrt{\alpha^2-(\beta+t_*)^2}\right),
    \end{equation}
    where
    \[
    \varphi = -\frac1\delta \log(1-\alpha') + \sqrt{\alpha^2 - \beta^2},
    \quad
    \psi = \sqrt{\varphi^2 - \alpha^2 + \beta^2},
    \quad
    t_* = \frac{(\alpha^2 - \beta^2) \psi}{\alpha\varphi + \beta\psi}.
    \]
\end{Theorem}

\begin{Remark}\label{rem:nig}
\begin{enumerate}[label=(\roman*)]
    \item We have excluded the case $\alpha = \beta = 0$, because in this case the moment-generating function is not defined for any $t\ne0$.
    \item The parameter $\psi$ is well defined, since
    \[
    \varphi^2 - \alpha^2 + \beta^2 = 
    \left(\tfrac1\delta \log(1-\alpha')\right)^2
    -\tfrac2\delta \log(1-\alpha') \sqrt{\alpha^2 - \beta^2} \ge0.
    \]
    \item Evidently, $\varphi \ge \psi \ge0$.
    \item\label{nig-iii} $t_* \in [0,\alpha-\beta]$, because
    $t_* = (\alpha - \beta)\frac{\alpha\psi +\beta \psi}{\alpha\varphi + \beta\psi}$
    and $\frac{\alpha\psi +\beta \psi}{\alpha\varphi + \beta\psi} \in [0,1]$ for $\varphi \ge \psi \ge0$.
    \item It follows from \ref{nig-iii} that \eqref{eq:evar-nig} is well defined.
\end{enumerate}
\end{Remark}

\begin{proof}
The corresponding moment-generating function is equal to
\begin{equation}
    m_X(t) = \exp\left\{\mu t+\delta \sqrt{\alpha^2 - \beta^2} - \delta \sqrt{\alpha^2 - (\beta + t)^2}\right\},
    \quad t \in [-\alpha-\beta, \alpha-\beta],
\end{equation}
see, e.g., \cite{barndorff1997normal}.
Then we get the following optimization problem:
\begin{equation}
    \EVaR_{\alpha'}(X) = \inf_{t \in (0, \alpha-\beta]} f(t),\nonumber
\end{equation}
where
\begin{align}
f(t) &\coloneqq \frac{1}{t}\left(\log\frac{1}{1-\alpha'} + \mu t+\delta \sqrt{\alpha^2 - \beta^2} - \delta \sqrt{\alpha^2 - (\beta + t)^2}\right)
\notag\\
&= \mu + \frac{\delta}{t} \left(\varphi - \sqrt{\alpha^2-(\beta+t)^2}\right).
\label{eq:nig-f}
\end{align}
Then we find the minimum point for positive $t$
\begin{align*}
f'(t) &= \frac{\delta}{t^2} \left( \frac{ (\beta+t) t}{\sqrt{\alpha^2-(\beta+t)^2}} -\varphi + \sqrt{\alpha^2-(\beta+t)^2}\right)
\\
&= \frac{\delta \left( (\beta+t) t -\varphi \sqrt{\alpha^2-(\beta+t)^2} + \alpha^2-(\beta+t)^2\right)}{t^2{\sqrt{\alpha^2-(\beta+t)^2}}} 
\\
&= \frac{\delta \left(\alpha^2-\beta^2 - \beta t -\varphi \sqrt{\alpha^2-(\beta+t)^2}\right)}{t^2{\sqrt{\alpha^2-(\beta+t)^2}}} 
\\
&= \frac{\delta \left((\alpha^2-\beta^2 - \beta t)^2 -\varphi^2 (\alpha^2-(\beta+t)^2)\right)}{t^2{\sqrt{\alpha^2-(\beta+t)^2}} \left(\alpha^2-\beta^2 - \beta t +\varphi \sqrt{\alpha^2-(\beta+t)^2}\right)} 
\end{align*}
Note that the denominator is positive for all $t\in[0,\alpha-\beta]$.
Indeed,
$\alpha^2-\beta^2 - \beta t \ge \alpha^2-\beta^2 - \beta (\alpha - \beta)
=\alpha (\alpha - \beta)\ge0$.
Hence, the sign of $f'(t)$ coincides with that of the function
\begin{align*}
g(t) &= (\alpha^2-\beta^2 - \beta t)^2 -\varphi^2 (\alpha^2-(\beta+t)^2)
\\
&= t^2 \left(\varphi^2 + \beta^2\right) + 2\beta \left(\varphi^2 - \alpha^2 + \beta^2\right) t - \left(\alpha^2 - \beta^2\right)\left(\varphi^2 - \alpha^2 + \beta^2\right)
\\
&= t^2 \left(\varphi^2 + \beta^2\right) + 2\beta \psi^2 t - \left(\alpha^2 - \beta^2\right)\psi^2.
\end{align*}
It is not hard to see that the quadratic equation $g(t) = 0$ has two solutions
\[
t = \frac{\psi(-\beta \psi \pm \alpha\phi)}{\varphi^2 + \beta^2},
\]
however, the solution with minus sign is obviously negative.
Let us consider the solution
\[
t_* = \frac{\psi(-\beta \psi +\alpha\phi)}{\varphi^2 + \beta^2}
= \frac{\psi(\alpha^2\phi^2 - \beta^2 \psi^2 )}{(\varphi^2 + \beta^2)(\beta \psi +\alpha\phi)}
= \frac{(\alpha^2 - \beta^2) \psi}{\alpha\varphi + \beta\psi},
\]
where we have used the relation
\[
\alpha^2\phi^2 - \beta^2 \psi^2 
= \alpha^2\phi^2 - \beta^2 (\phi^2 - \alpha^2 + \beta^2)
= \left(\alpha^2 - \beta^2\right) \left(\phi^2 + \beta^2\right)
\]
According to Remark \ref{rem:nig} \ref{nig-iii}, $t_*\in[0,\alpha-\beta]$.
Moreover, from the properties of the quadratic function we get that $g(t) < 0$ for $t\in[0, t_*)$ and  $g(t) > 0$ for $t\in(t_*,\alpha-\beta]$, and the derivative $f'(t)$ demonstrates the same behavior.
This means that $t_*$ is the minimum point of $f(t)$ on $[0,\alpha-\beta]$.
Substituting this value into \eqref{eq:nig-f}, we obtain \eqref{eq:evar-nig}.
\end{proof}
\begin{figure}
\centering
    \includegraphics[scale=0.3]{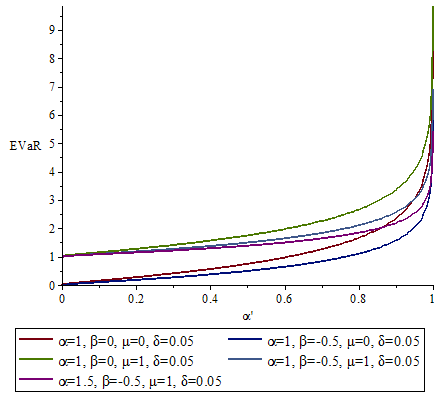}
    \includegraphics[scale=0.35]{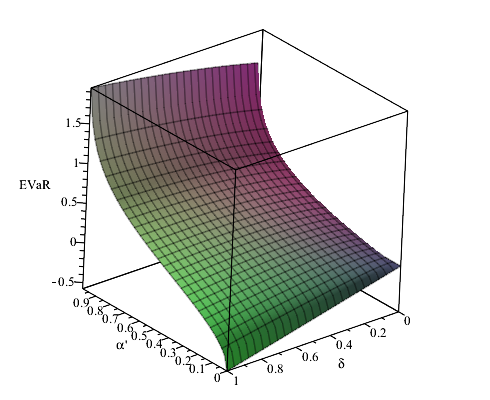}\\
    \includegraphics[scale=0.35]{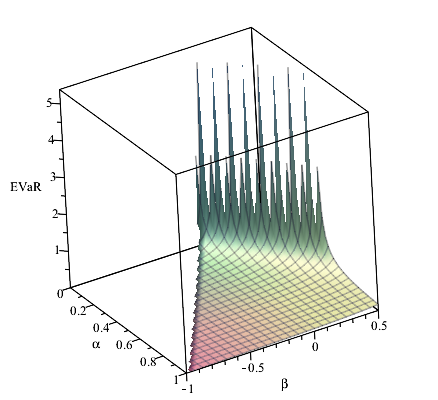}
    \includegraphics[scale=0.35]{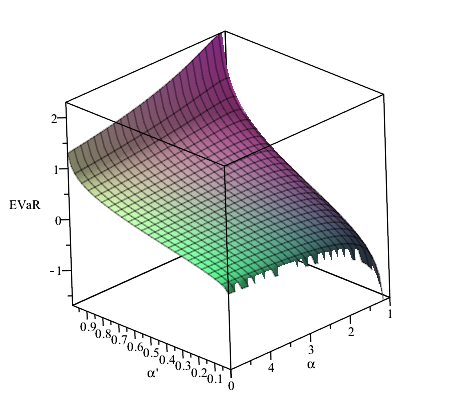}
    \caption{
        EVaR for NIG distribution. For the second figure $\alpha = 2, \beta=-1, \mu =0$, for the third figure $\alpha' = 0.05, \mu =0, \delta=0.05$, for the last figure $\beta = -1$, $\mu =0$, $\delta=1$.}
\end{figure}

\section{Concluding Remarks}
  There are not many models for EVaR described in the current scientific literature. Most often, only the standard normal distribution model is considered. The lack of extended  results in this area is due to the complexity of calculations of EVaR for the majority of the main distributions. Previously, only normal distribution was mostly considered. 

In this paper,   EVaR for several distributions was calculated via Lambert function.   The distributions that are most commonly used in statistical models and loss modelling were discussed, namely, Poisson, compound Poisson, gamma, exponential, chi-squared, Laplace,   inverse Gaussian and normal inverse Gaussian  distributions.

The result of this work brings the possibility of further studying the behaviour of EVaR for the already defined models, as well as expanding the models for which EVaR can be analytically calculated.

\appendix
\section[\appendixname~\thesection]{EVaR for normal distribution}
\label{sec:evar-normal}
Let us, for the reader's convenience,  provide the formula  of EVaR
for the normal distribution from \cite{Ah12}. Let us formulate the corresponding lemma with a brief proof.
\begin{Lemma}\label{th:evar-normal}
    Let $X \sim \Normal(\mu, \sigma^2)$ be a normal distribution. Then $\forall \alpha \in (0,1):$
    \begin{equation}\label{eq:evar-norm}
        \EVaR_\alpha(X) = \mu+\sigma\sqrt{-2\log(1-\alpha)}.
    \end{equation}
\end{Lemma}
\begin{proof}
    The moment-generating function of the normal distribution is given by $m_X(t) = \exp\{\mu t+\frac{1}{2}t^2\sigma^2\}$.
    We substitute it into the EVaR definition \eqref{eq:evar-def} and obtain:
    \begin{equation}\label{eq:evar-norm1}
        \EVaR_\alpha(X) = \inf_{t > 0}\left(-\frac{1}{t}\log(1-\alpha) + \mu+\frac{1}{2}t\sigma^2\right).
    \end{equation}
    It is not hard to show by differentiation that the infimum is attained at the point
    \[
    t^* := \frac{1}{\sigma}\sqrt{-2\log(1-\alpha)}.
    \]
    After substituting $t^*$ into \eqref{eq:evar-norm1}, we get \eqref{eq:evar-norm}.
\end{proof}

\section[\appendixname~\thesection]{List of related probability distributions with density functions}
\label{sec:pdfs}
\begin{enumerate}
    \item {Gamma distribution $G(k,\theta)$: $f(x) = \frac{x^{k-1}e^{-x/\theta}}{\Gamma(k)\theta^k}, x>0, k>0, \theta>0.$}
    \item {Exponential distribution $\Exp(\lambda)$: $f(x) = \lambda e^{-\lambda x}, x \geq 0, \lambda>0. $}
    \item {Chi-squared distribution $\chi^2(k)$: $f(x) = \frac{x^{k/2-1}e^{-x/2}}{\Gamma(k/2)2^{k/2}}, x>0, k\in \mathbb{N} .$}
    \item {Laplace distribution $L(\mu, b)$: $f(x) = \frac{1}{2b}\exp\left\{-\frac{|x-\mu|}{b}\right\}, x, \mu \in \mathbb{R}, b>0.$}
    \item{Inverse Gaussian distribution $\IG(\mu, \lambda)$: $f(x) = \sqrt{\frac{\lambda}{2\pi x^3}\exp \left\{ -\frac{\lambda(x-\mu)^2}{2\mu^2x}\right\}},\\ x>0,\mu>0,\lambda>0.$}
    \item{Normal inverse Gaussian distribution $\NIG(\alpha,\beta,\mu,\delta)$: 
    \[
    f(x) = \frac{\alpha\delta K_1(\alpha \sqrt{\delta^2+(x-\mu)^2})}{\pi \sqrt{\delta^2+(x-\mu)^2}} \exp\left\{ \delta \sqrt{\alpha^2-\beta^2 } + \beta(x-\mu)\right\},
    \]
    $K_1(x) = \frac{1}{2}\int_0^{+\infty}\exp \left\{ -\frac{1}{2}x(t+t^{-1})\right\}dt$  is a   modified  Bessel function  of the  third  kind with index 1, $x,\beta,\mu, \in \mathbb{R},\alpha \geq |\beta|,\delta \geq 0.$} 
\end{enumerate}

\bibliographystyle{abbrv}
\bibliography{biblio}

\end{document}